\documentclass[draftcls,onecolumn,11pt]{IEEEtran}

\usepackage{amsmath,amstext,amssymb,epsf}
\usepackage{fullpage,latexsym}
\usepackage{epsfig}
\usepackage{setspace}
\numberwithin{equation}{section} 

 \newtheorem{lemma}{Lemma}[section]

 \newtheorem{definition}[lemma]{Definition}
 \newtheorem{rem}[lemma]{Remark}
\newenvironment{remark}{\begin{rem}}{\hspace*{\fill}$\diamondsuit$\end{rem}}
 \newtheorem{ex}[lemma]{Example}
\newenvironment{example}{\begin{ex}}{\hspace*{\fill}$\diamondsuit$\end{ex}}

\newcommand{\K}{K}

\begin{document}

\title{On Empirical Entropy}
\author{Paul M.B. Vit\'{a}nyi
\thanks{
Paul Vit\'{a}nyi is a CWI Fellow with the National Research Center for 
Mathematics and Computer Science in the Netherlands (CWI)
and Emeritus Professor of Computer Science at the University of Amsterdam.
He was supported in part by the
the ESF QiT Programmme, the EU NoE PASCAL II.
Address:
CWI, Science Park 123,
1098 XG Amsterdam, The Netherlands.
Email: {\tt Paul.Vitanyi@cwi.nl}.
}}


\maketitle

\begin{abstract}
We propose a compression-based version of the empirical entropy 
of a finite string over
a finite alphabet. Whereas previously one 
considers the naked entropy of (possibly
higher order) Markov processes, we consider
the sum of the description of the random variable 
involved plus the entropy it induces.
We assume only that the distribution involved is computable.
To test the new notion we compare the Normalized Information Distance
(the similarity metric) with a related measure based on Mutual
Information in Shannon's framework. This way the similarities and differences
of the last two concepts are exposed.

\end{abstract}

\section{Introduction}
\label{sect.intro}
In the basic set-up of  Shannon~\cite{Sh48} 
a message is a finite string over a finite alphabet.
One is interested in the expected number of bits to transmit a message
from a sender to a receiver, when both the sender and the receiver
consider the same ensemble of messages (the set of possible messages
provided with a probability for each message). 
The expected number of bits 
is known as the entropy of the ensemble of messages.
This ensemble is also known as the source.
  
The empirical entropy of a single message is taken to
be the entropy of a source that produced it as a typical element. 
(The notion of ``typicality'' is defined differently by different
authors and we take here the intuitive meaning.)
Traditionally, this source
is a (possibly higher order) Markov process. This leads to
the definition in Example~\ref{exam.entropy}. Here we want
to liberate the notion so that it encompasses 
all computable random variables with
finitely many outcomes consisting of finite strings over a finite alphabet.
Moreover, since we are given only a single
message, but not the ensemble from which it is an element, the
new empirical entropy should provide both this ensemble and the
entropy it induces. If we are given just the entropy but not the
ensemble involved, then a receiver cannot in general reconstruct the message.
Moreover, we are given a single message which has a particular length, say $n$.
Therefore, given the family of random variables we draw upon,
we can select one of them and compute the probability of every
message of length $n$. For fixed $n$, this results in a Bernoulli
variable 
that has $|\Sigma|^n$ outcomes.

We are thus led to a notion of empirical entropy
that consists of a description of the Bernoulli variable
involved plus the related entropy of the message induced.
Since we assume the original probability mass function to be computable,
the Bernoulli variable is computable and
its effective description length can be expressed
by its Kolmogorov complexity. 

Normalized Information Distance (explained below) between two
finite objects is often confused
with a similar distance between two random variables. The last distance
is expressed in terms of probabilistic mutual information.
We use our new notion to explain the differences between 
the former distance between two individual objects and the 
latter distance between
two random variables. This difference parallels that 
between the Kolmogorov complexity of a single
finite object
and the entropy of
a random variable.
The former quantifies the information 
in a finite object, while the latter gives us the expected number
of bits to communicate any outcome of a random variable known
to both the sender and the receiver. 
Computability notions are reviewed in Appendix~\ref{sect.computability},
and Kolmogorov complexity in Appendix~\ref{sect.kolmcompl}.

\subsection{preliminaries}
We write {\em string} to mean a finite string over a finite alphabet $\Sigma$.
  Other finite objects can be encoded into strings in natural
ways.  The set of strings is denoted by $\Sigma^*$. We usually take
$\Sigma = \{0,1\}$. The {\em length}
of a string $x$ is the number of letters in $\Sigma$ in it denoted as $|x|$. 
The {\em empty}
string $\epsilon$ has length $|\epsilon| = 0$.
Identify the natural numbers
${\cal N}$ (including 0) and $\{0,1\}^*$ according to the
correspondence
 \begin{equation}\label{order}
 (0, \epsilon ), (1,0), (2,1), (3,00), (4,01), \ldots . 
 \end{equation}
Then, $|010|=3$.
The emphasis here is on binary sequences only for convenience;
observations in every finite alphabet can be so encoded in a way
that is `theory neutral.' For example, if a finite alphabet $\Sigma$ has
cardinality $2^k$, then every element $i \in \Sigma$ can be encoded
by $\sigma(i)$ which
is a block of bits of length $k$. With this encoding every $x \in \Sigma^*$
satisfies that the Kolmogorov complexity
$\K(x)=\K(\sigma(x))$ (see Appendix~\ref{sect.kolmcompl} for basic definitions
and results on Kolmogorov complexity)
up to an additive constant that is
independent of $x$.

\section{The New Empirical Entropy}

Let $X$ be a random variable with outcomes in a finite alphabet ${\bf X}$.
Shannon's entropy \cite{Sh48} is 
\[
H(X) = \sum_{x \in {\cal X}} P(X=x) \log 1/P(X=x) .
\]
There are three items involved in the
new empirical entropy of data $x$:
\begin{itemize}
\item
A class of random variables
like the set of Bernoulli processes, or the set of higher
order  Markov
processes; from each element of this class we construct
a Bernoulli variable $X$ with $|\Sigma|^n$ outcomes
of length $n$;
\item
a selection of a random variable from this Bernoulli class such that $x$ 
is a typical outcome, and
\item
a description of this random variable plus its entropy.
\end{itemize}
This is reminiscent of
universal coding essentially
due to Kolmogorov~\cite{Ko65}, and of 
two-part MDL
due to Rissanen~\cite{Ri89}.
In its simplest form the former, assuming
a Bernoulli process, codes a
string $x$ of length $n$ over a finite alphabet $\Sigma$ 
as follows: A string containing a description of $n,|\Sigma|$ and $n/n_i$ 
($1 \leq i \leq |\Sigma|$), and the index of $x$ in the set constrained by
these items.
The coding should be such that the individual substrings can be parsed,
except the description of the index which we put last. 
This takes additive terms that are logarithmic in the length of
the items except the last one.
The universal code takes
$O(|\Sigma|\log n)+{{n} \choose {n/n_1 \cdots n/n_{|\Sigma|}}}$ bits.
The two-part MDL complexity of a string \cite{Ri89},
is the minimum of the self-information of that string with respect to
a source and the number of bits needed to represent that source.
The source is not required to be Markovian
and the two-part MDL takes into account its complexity. However,
the methods of encoding are arbitrary. 

An $n$-length outcome $x=x_1, x_2, \ldots, x_n$ over $\Sigma$ 
is the outcome of a stochastic process 
$X_1, X_2, \ldots , X_n$ characterized by a joint probability
mass function $\Pr(\{X_1, X_2, \ldots , X_n)=(x_1,x_2, \ldots , x_n)\}$.
For technical reasons we replace the list $X_1,X_2, \ldots , X_n$ by a single
Bernoulli random variable $X$ with outcomes in 
${\bf X}=\Sigma^n$. Here, the random variables $X_i$ may be independent
copies of a single random variable as is the case wen the source
stochastic process is a Bernoulli variable. But the source stochastic
process may be a higher order Markov chain making some or all
$X_i$s dependent (this depends on whether the order of the Markov chain
is greater then $n$). For certain stochastic processes all $X_i$s
are dependent for every $n$: the stochastic process
assigns a probability to every outcome in $\Sigma^*$. 

\begin{definition}\label{def.ee}
\rm
Let $n$ be an integer, $\Sigma$ a finite alphabet, $x \in \Sigma^n$ be a string,
${\cal X}$ a family of computable
processes, each process $\Xi \in {\cal X}$
producing (possibly by repetition) a sequence of
(possibly dependent) random variables
$X=X_1, X_2, \ldots , X_n$, with
$\Pr(X=x)$ is computable and 
$H(X) < \infty$. 
The {\em empirical entropy}
of $x$ with respect to ${\cal X}$ is given by 
\[
H({x} | {\cal X})= \min_{\Xi \in {\cal X}}  \{K(X)+H(X): 
|H(X) - \log 1/\Pr(X=x)| \; {\rm is} \; {\rm minimal} 
\}.
\] 
\end{definition}

This means that the expected binary length of encoding an outcome
of $X$ is as close as possible to $\log 1/\Pr(X=x)$. 
In the two-part description the complexity
part describes $X$,
and the entropy part is the ignorance about the data $x$ in the set $\Sigma^n$
given~$X$. 
\begin{remark}
\rm
By assumption $n$ is fixed. By Theorem 3 in
\cite{Sh48}, i.e. the asymptotic equidistribution property, 
for ergodic Markov sources the following is the case.  
Let $H$ be the per symbol entropy of the source. For example, if the source
$\Xi$ is Bernoulli with $\Pr(\Xi = s_i)=p(s_i)$ ($s_i \in \Sigma$
for $1 \leq i \leq |\Sigma|$), then 
$H= \sum_{i=1}^{|\Sigma|} p(s_1) \log 1/p(s_i)$. Let $X$ 
be the induced Bernoulli
variable with $|\Sigma|^n$ outcomes consisting of sequences of 
length $n$ over $\Sigma$. Then,
for every $\epsilon, \delta > 0$ there is an $n_0$ such that
the sequences of length $n \geq n_0$ are divided into two classes:
one set with total probability less than $\epsilon$ and one
set such that for every $y$ in this set holds 
$|H - \frac{1}{n}\log 1/\Pr(X=y)| < \delta$. Note that $H(X)=nH$.
Thus, for large enough $n$ we are almost
certain to have $|H(X) - \log 1/\Pr(X=x)|=o(n)$.   

%
%

Set $\epsilon = \delta$ for convenience. 
We call the set of $y$'s such that $|H(X)-\log 1/\Pr(X=y)|= \epsilon n$, 
with $\epsilon > 0$ and some $n_0$ depending on $\epsilon$ and $n\geq n_0$,
the $\epsilon$-{\em typical} outcomes of $X$.
The cardinality of the set $S \subseteq \Sigma^n$ of such $y$'s satisfies
\[
(1-\epsilon) |\Sigma|^{H(X)- \epsilon n} \leq |S| \leq |\Sigma|^{H(X)+\epsilon n}.  
\]
See \cite{CT91} Theorem 3.1.2.
\end{remark}
\begin{lemma}
Assume Definition~\ref{def.ee}. Then, 
$K(X) \leq K(x, {\cal X})+O(1)$.
\end{lemma}
\begin{proof}
The family ${\cal X}$ consists of computable random variables,
that is, in essence of computable probability mass functions. 
The family of all lower semicomputable semiprobability
mass functions can be effectively enumerated, possibly with repetitions,
Theorem 4.3.1 in \cite{LiVi97}.
The latter family contains all computable probability mass functions, 
hence it contains $X$. Thus, if we know $x, {\cal X}$
we can compute the 
$X \in {\cal X}$ of Definition~\ref{def.ee} by going through this list.
\end{proof}

\begin{example}\label{exam.entropy}
\rm
Assume Definition~\ref{def.ee}. Let $n_i$ be the number of occurrences of
the $i$th character of $\Sigma$ in $x$. If $w$ is a string then $x_w$ is
the string obtained by concatenating the characters immediately
following occurrences of $w$ in $x$. The cardinality 
$|x_w|$ is the number of occurrences of $w$ in $x$ 
unless $w$ occurs as a suffix of $x$ in which case it is 1 less. 
In \cite{KM99,Ma01,Ga06} the $k$th order empirical entropy of $x$ is defined by
\begin{equation}\label{eq.eentropy}
H_k(x)= \left\{ \begin{array}{ll}
\frac{1}{n} \sum_{i=1}^{|\Sigma|} n_i \log \frac{n}{n_i} &{\rm for }\;\; k=0, \\
\frac{1}{n} \sum_{|w|=k} |x_w| H_0(x_w) & {\rm for }\;\; k > 0.
\end{array}
\right.
\end{equation}
The $k$th order empirical entropy of $x$ can be reconstructed from $x$ 
once we know $k$. 
The $k$th order
empirical entropy of $x$ results from
the probability induced by a $k$th order Markov source $\Xi \in {\cal X}$. 
(A Bernoulli
process is a $0$th order Markov source.)

Let ${\cal X}$ to be the family of 
$k$th order Markov sources (a specific $k \geq 0$), 
provided the transition probabilities
are computable. Such a family is subsumed under
Definition~\ref{def.ee}. 
Let $x$ be a string over 
$\Sigma$ which is typically produced by such a Markov source of order $k$.
The empirical entropy $H(x|{\cal X})$ of $x$ is
$K(X)+nH_k(x)$. Here $X$ is the random variable 
associated with the $k$th order empirical entropy
computed from $x$. Note that the empirical entropy $H_k(x)$
stops being a reasonable complexity metric for almost all strings 
roughly when $|\Sigma|^k$ surpasses $n$, \cite{Ga06}.
\end{example}
\begin{example}
\rm
Let $x= (10)^{n/2}$ for even $n$ (that is,
$n/2$ copies of the pattern "$10$").  Let ${\cal X}_1$ be
the family of binary Bernoulli processes. 
The empirical entropy $H(x|{\cal X}_1)$
is reached for i.i.d. sequence $X=X_1,X_2, \ldots ,X_n \in {\cal X}_1$, 
each $X_i$ being a copy of
the same random variable $Y$ with outcomes in $\{0,1\}$
with $P(Y=1)= \frac{1}{2}$. 
Then, $H({x} | {\cal X}_1)= K(X) + n H(Y)$.
Then $X$ can be computed from the information concerning $n$ in $O(\log n)$
bits, the particular $\Xi \in {\cal X}$ used in $O(1)$ bits, and 
a program of $O(1)$ bits to compute $X$ from this information.
In this way $K(X)=O(\log n)$. Moreover, $H(Y)=1$, so that 
$H({x} | {\cal X}_1) = n+O(\log n)$.

Let ${\cal X}_2$
be the family of first order Markov 
processes with $2$ transitions each and 
with output alphabet 
$\{0,1\}$ for each state. 
The empirical entropy $H(x|{\cal X}_2)$ is reached
for the $n$-bit output of a 
deterministic ``parity'' Markov process. That is, 
$X=X_1,X_2, \ldots , X_{n}$ and every $X_i$ gives the output at time $i$
of the Markov process 
with 2 states $s_0$ and $s_1$ defined as follows.
The transit probabilities are $p(s_0 \rightarrow s_1)=1$
and $p(s_1 \rightarrow s_0)=1$, while the output in state $s_0$ is $0$
and in state $s_1$ is $1$. The start state is $s_0$. In this way, 
$P(X=(10)^{n/2})=1$ while $H(X)=0$.
Then, $H({x} | {\cal X}_2)= K(X) + 
H(X)$.
Here $K(X)=O(\log n)$, since we require a description of $n$, the 2-state
Merkov process involved, and a program to compute $X$ from this information.
Since the outcome is deterministic, $H(X)=0$, so that 
$H({x} | {\cal X}_2)=O(\log n)$.
\end{example}
\begin{example}
\rm
Consider the first
$n$ bits of $\pi = 3.1415 \ldots .$
Let ${\cal X}_1$ be the family of 
Bernoulli processes. 
Empirically, it has been established that the frequency
of $1$'s in the binary expansion of $\pi$  is $n/2 \pm O(\sqrt{n})$, 
that is, the binary expaqnsion of $\pi$ is a typical 
pseudorandom sequence. Hence,
$H({x} | {\cal X}_1)= K(X) + n H(X)$ where $X=X_1, X_2, \ldots , X_n \in {\cal X}_1$ and the $X_i$'s
are $n$ i.i.d. distributed copies of  $Y$. Here $Y$  is
a Bernoulli process with $P(Y=1)=\frac{1}{2}$.
Then $K(X)=O(\log n)$ and $H(Y)=1$, so that $
H({x} | {\cal X}_1)=n + O(\log n)$.

Let ${\cal X}_2$ be the family of computable random variables with as outcomes
binary strings of length $n$. 
We know that there is a small program, say of about $10,000$
bits, incorporating an approximation
algorithm that generates the successive bits of $\pi$ forever.
Telling it to stop after $n$ bits, we can generate the computable
Bernoulli variable $X \in {\cal X}_2$
assigning probability 1 to $x$ and probability 0 to
any other binary string of length $n$. Assume $n=1,000,000,000$. 
Then, we have $K(X)\leq \log 1,000,000,000 +c \approx 30+c$
where the $c$ additive term is the number of bits of 
the program
to compute $\pi$ and a program required to
turn the logarithmic description of $1,000,000,000$ and the program 
to compute $\pi$ into the random variable $X$. Finally,    
$H(X)=0$. Therefore, $H({x} | {\cal X}_2)\leq 10,030+c$.
\end{example}
\begin{example}
\rm
Consider printed English, say just lower case and 
space signs, ignoring 
the other signs. 
The entropy of representative examples of 
printed English has been estimated experimentally by 
Shannon \cite{Sh51}
based on human subjects guesses of successive characters in a text.
His estimate is between 0.6 and 1.3 bits per character (bpc), and
\cite{TC96} obtained an estimate of 1.46 bpc
for PPM based models, which we will use in this example.
PPM (prediction by partial matching)
 is an adaptive statistical data compression technique.
It is based on context modeling and prediction and uses 
a set of previous symbols in the uncompressed symbol stream 
to predict the next symbol in the stream, rather like a mechanical
version of Shannon's method.
Consider a text of $n$ characters over the alphabet used
by \cite{TC96}, and let ${\cal X}$ be the
class of PPM based models with $n$ output characters over the 
used alphabet. Since the PPM
machine can be described in $O(1)$ bits (its program is finite)
and the length $n$ in $O(\log n)$ bits, we have $K(X)=O(\log n)$.
Hence, $H({x} | {\cal X})\leq K(X) + 1.46n = 1.46n +O(\log n)$.
\end{example}

In these examples we see that the empirical entropy is higher when
the family of random variables considered is simpler. For simple
random variables
the knowledge in the Kolmogorov complexity part is neglible.
The empirical entropy with respect to a complex family of
random variables can be lower than that with respect to a family of simple
random variables
by transforming the ignorance in the entropy part into
knowledge in the Kolmogorov complexity part. We use this observation
to consider the widest family of  computable probability mass functions. 

\begin{lemma}\label{lem.1}
Let ${\cal X}$ be the family of computable random variables
$X$ with $H(X)<\infty$, and
$x \in \Sigma^*$ with $|\Sigma| < \infty$.
Then, $H(x|{\cal X}) = K(x)+O(1)$.
\end{lemma}
\begin{proof}
First, let $p_x$ be a shortest prefix program which computes $x$. Hence $|p_x| = K(x)$.
By adding $O(1)$
bits to it we have a program $p_p$ 
which computes a probability mass function $p$
with $p(x)=1$ and $p(y)=0$ for $y \neq x$ ($x,y \in \Sigma^*$).  
Hence $|p_p| \leq K(x)+O(1)$. 

Second, let $q_p$ be a shortest prefix program which computes a probability 
mass function $p$
with $p(x)=1$ and $p(y)=0$ for $y \neq x$ ($x,y \in \Sigma^*$).
Thus, $|q_p| \leq |p_p|$.
Adding $O(1)$ bits to $q_p$ we have a program $q_x$ which computes $x$.
Then, $K(x) \leq |q_p|+O(1)$. 

Altogether, $|q_p| = K(x)+O(1)$.
\end{proof}

For the sequel of this paper, we need to extend the notion of empirical
entropy to joint probability mass functions.
\begin{definition}\label{def.ee2}
\rm
Let $n$ be an integer, $\Sigma$ a finite alphabet, 
$x,y \in \Sigma^n$ be strings,
${\cal Z}$ be the family of computable
joint probability mass functions,
$Z \in {\cal Z}$ 
and $(x,y)$ an outcome of $Z$. Let the probability
mass function $p(x,y) = P(Z=(x,y))$ have a finite joint entropy
$H(Z) < \infty$.
The {\em empirical entropy}
of $(x,y)$ with respect to ${\cal Z}$ is
\[
H(x,y | {\cal Z} )= \min_{Z \in {\cal Z}} 
 \{K(Z)+H(Z): 
|H(Z)- \log 1/p(x,y)|\;\;{\rm is}\;\;{\rm minimal}\}. 
\]
\end{definition}

\begin{lemma}\label{lem.2}
Let ${\cal Z}$ be  the family of computable joint probability 
mass functions $Z$ with $H(Z)<\infty$, and
$x,y \in \Sigma^*$ with $|\Sigma| < \infty$.
Then,
$H(x,y|{\cal Z}) = K(x,y)+O(1)$.
\end{lemma}
\begin{proof}
Similar to that of Lemma~\ref{lem.1}.
\end{proof}
 
\section{Normalized Information Distance}

The classical notion of
Kolmogorov complexity \cite{Ko65} is an objective measure
for the information in
a {\em single} object, and information distance measures the information
between a {\em pair} of objects \cite{BGLVZ98}. This last notion has
spawned research in the theoretical direction, see the many Google Scholar
citations to the above reference.
Research in the practical direction has focused on
the normalized information distance (NID),
also called ``the similarity metric,'' which arises
by normalizing the information distance in a proper manner.
(The NID is defined by \eqref{eq.nid} below.) 

If we 
approximate the Kolmogorov complexity through real-world
compressors \cite{Li03,CVW03,CV04},
then we obtain the
normalized compression distance (NCD) from the NID.
This is a parameter-free, feature-free, 
and alignment-free
similarity measure  that
has had great
impact in applications. (Only the compressor used can be viewed as
a parameter or feature.)
The NCD was preceded by a related
nonoptimal distance \cite{LBCKKZ01}.
In \cite{KLRWLH07} another variant of the NCD has been tested
on all major time-sequence databases used in all major data-mining conferences
against all other major methods used. The compression method turned out
to be competitive in general
and superior in
heterogeneous data clustering and anomaly detection.

There have been many applications in pattern recognition, phylogeny, clustering,
and classification, ranging from
hurricane forecasting and music to
to genomics and analysis of network traffic,
see the many
papers referencing
\cite{Li03,CVW03,CV04}
in Google Scholar.
In \cite{Li03} it is shown that the NID, and in \cite{CV04} 
that the NCD subject to mild conditions on the used compressor, are
metrics up to negligible discrepancies
in the metric (in)equalities and that they are always between 0 and 1.
The computability status of the NID has been resolved in \cite{TTV10}.
The NCD is computable by definition.

The {\em information distance} $D(x,y)$ between strings $x$ and $y$
is defined as
\[
D(x,y)= \min_{p} \{|p|: U(p,x)=y \wedge U(p,y)=x \},
\]
where $U$ is the reference universal Turing machine above.
Like the Kolmogorov complexity $K$, the distance function $D$
is upper semicomputable. Define
\[E(x,y)= \max \{K(x|y),K(y|x)\}.
\]
In \cite{BGLVZ98} it is shown that
the function $E$ is upper semicomputable,
$D(x,y)= E(x,y)+O(\log E(x,y))$, the function $E$ is a metric (more precisely,
that it satisfies the metric (in)equalities up to a constant),
and that $E$ is minimal (up to a constant) among all
upper semicomputable distance functions $D'$ satisfying the mild
normalization conditions $\sum_{y:y \neq x} 2^{-D'(x,y)} \leq 1$ and
$\sum_{x:x \neq y} 2^{-D'(x,y)} \leq 1$.
(Here and elsewhere in this paper ``$\log$'' denotes the binary logarithm.)
The {\em normalized information distance} (NID) $e$ is defined by
\begin{equation}
e(x,y) = \frac{E(x,y)}{\max\{K(x),K(y)\}}.
\end{equation}
It is straightforward that $0 \leq e(x,y) \leq 1$ up to some minor
discrepancies for all $x,y \in \{0,1\}^*$.
Rewriting $e$ using \eqref{eq.soi} yields
\begin{equation}
\label{eq.nid}
e(x,y) = \frac{\K(x,y) - \min\{\K(x),\K(y)\}}{\max\{K(x),K(y)\}},
\end{equation}
up to some lower order terms that we ignore.

\begin{lemma}\label{lem.3}
Let $x$ be a string, ${\cal X}$, ${\cal Z}$ be
the families of random variables with 
computable probability mass functions and
computable joint probability mass functions, respectively. 
Moreover, for $X\in{\cal X}$ and $Z \in {\cal Z}$ we have $H(X),H(Z) < \infty$.
Then, we can substitute
the Kolmogorov complexities in \eqref{eq.nid} by the corresponding
empirical entropies as in \eqref{eq.nidee}.
\end{lemma}
\begin{proof}
By Lemma's~\ref{lem.1}
and \ref{lem.2} we know the following. For ${\cal X}$ is the family
of computable probability mass functions, $H(x|{\cal X})=K(x)$,
$H(y|{\cal X})=K(y)$. For ${\cal Z}$ is the family of computable joint
probability mass functions, $H(x,y|{\cal Z})=K(x,y)$. Hence,
\begin{equation}\label{eq.nidee}
e(x,y) = \frac{H(x,y|{\cal Z}) - 
\min\{H(x|{\cal X}),H(y|{\cal X})\}}{\max\{H(x|{\cal X}),H(y|{\cal X})\}},
\end{equation}
ignoring lower order terms.
\end{proof}

\begin{remark}
\rm
In Lemma~\ref{lem.3} we can replace the computable random variables
by the restriction to computable random variables that have
a singleton support, that is, probability mass functions $p$ with
$p(x)=1$ for some $x$ and $p(y)=0$ for all $y \neq x$.
Alternatively, we can replace it by the family of computable
Markov processes. To see this, for every $x$ of length $n$ there is
a computable Markov process $M$ of order $n-1$ that outputs
$x$ deterministically and $K(x)=K(M)+O(1)$.

Clearly, if we replace the family of computable probability mass functions
in the empirical entropies in Lemma~\ref{lem.3} 
by weaker subfamilies like the families based on computable
Bernoulli functions, computable Gaussians, or computable first order
Markov processes, then Lemma~\ref{lem.3} will not hold in general.  
\end{remark}

\begin{remark}
\rm
The NCD is defined by
\begin{equation}
\label{eq.ncd}
NCD_Z(x,y) = \frac{|Z(xy)| - \min\{|Z(x)|,|Z(y)|\}}{\max\{|Z(x)|,|Z(y)|\}},
\end{equation}
where $Z(x)$ is the compressed version of $x$ when it
is compressed by a lossless compressor $Z$. 
We have substituted $xy$ for the pair $(x,y)$ both for convenience and with
ignorable consequences.
Consider a simple compressor that uses only Bernoulli variables, for example
a Huffman code compressor. The compressed version of a string is preceded by
a header containing information identifying the compressor and the 
charcteristics used (the relative frequencies in this case) to compress
the source string. 
In general this is the case with every compressor.
(In \cite{Ci07} the NCD based on compressors computing
the static Huffman code of a Bernoulli variable
 is shown to be the total Kullback-Leibler divergence
to the mean. We refrain from explaining these terms since are extraneous
to our treatment.) 

Thus, $Z(x)$ is comprised of the header generated by $Z$ for $x$. This
header makes it possible to use the uncompress feature, denoted here by
$Z^{-1}$ so that $Z^{-1}Z(x)=x$. The header describes a
random variable $\Xi$ based on the compressor $Z$.
The family of random variables
induced by the compressor $Z$ can be denoted by ${\cal X}_Z$.

In this way, we can define the Bernoulli variable $X$ used to compress $x$. 
The empirical entropy $H(x|{\cal X}_Z)=K(X)+H(X)$. Here $K(X)$ is uncomputable.
We approximate it by the length of the header, say $|{\rm \alpha}(X)|$. 
The Bernoulli
variable $X$ has entropy $H(X)$ and $|Z(x)|=|{\rm \alpha}(X)|+H(X)$. 
Similarly for $y$ and $(x,y)$.
Therefore,
\begin{equation}\label{eq.ncdZ}
NCD_Z(x,y) = \frac{|{\rm \alpha}(XY)|+H(X,Y) - 
\min\{|{\rm \alpha}(X)|+H(X), |{\rm \alpha}(Y)|+H(Y)\}}
{\max\{|{\rm \alpha}(X)|+H(X),|{\rm \alpha}(Y)|+H(Y)\}},
\end{equation}
ignoring lower order terms, where $|{\rm \alpha}(X)| \geq K(X)$, 
$|{\rm \alpha}(Y)| \geq K(Y)$,
and $|{\rm \alpha}(XY)| \geq K(XY)$.

\end{remark}

\section{Mutual Information}

In \cite{YJ01,BCL02,KSAG05,DHHM05,YMZA07,KG09} 
the entropy and joint entropy 
of a pair of sequences is determined, and this is directly equated
with the Kolmogorov complexity of those sequences. The Shannon type
probabilistic version of \eqref{eq.nid} is
\begin{eqnarray*}
e_H(X,Y) &=& \frac{H(X,Y) - \min\{H(X),H(Y)\}}{\max\{H(X),H(Y)\}}
\\&= &1-\frac{\max\{H(X),H(Y)\} - H(X,Y) + 
\min\{H(X),H(Y)\}}{\max\{H(X),H(Y)\}}
\\&= &1-\frac{I(X;Y)}{\max\{H(X),H(Y)\}}, 
\end{eqnarray*}
since the {\em mutual information} $I(X;Y)$ between random variables
$X$ and $Y$ is
\[
I(X;Y) = H(X)+H(Y) - H(X,Y),
\]
and
\[
\max\{H(X),H(Y)\} + \min\{H(X),H(Y)\} = H(X)+H(Y).
\]
In this way, $e_H(X,Y)$ is 1 minus the mutual information between
random variables $X$ and $Y$ per bit of the 
maximal entropy. How do the cited references
 connect this distance between two random variables
to \eqref{eq.nid}, the distance between two individual outcomes
$x$ and $y$? 

Ostensibly one has to replace the entropy of random variables $X$ and $Y$
by the empirical entropy according to Definition~\ref{def.ee}
deduced from strings $x$ and $y$. To obtain the required result
\eqref{eq.nid} one has to use 
families ${\cal X}$, ${\cal Y}$, ${\cal Z}$ 
of computable random variables such that $K(x)=H(x|{\cal X})$,
$K(y)=H(y|{\cal Y})$, and $K(x,y)=H(x,y|{\cal Z})$.
In our framework this is possible  only if ${\cal X},{\cal Y}$ are 
appropriate families
of computable random variables, and ${\cal Z}$ is an appropriate
 family of computable
joint random variables. 
Outside our framework the widest notion of empirical 
entropy is \eqref{eq.eentropy} and there it is not possible at all.

To obtain computable approximations using a real-world compressor $Z$
for $x$ and $y$ as in \eqref{eq.ncd} we can take the empirical entropy
based on compressor $Z$ as in \eqref{eq.ncd} and \eqref{eq.ncdZ}.

\appendix
\subsection{Computability}\label{sect.computability}
In 1936 A.M. Turing \cite{Tu36} defined the hypothetical `Turing machine'
whose computations are
intended to give an operational and formal definition
of the intuitive notion of computability in the discrete domain.
These Turing machines compute integer functions,
the {\em computable} functions. By using pairs of integers for the
arguments and values we can extend computable functions
to functions with rational arguments and/or values.
The notion of computability can be further
extended, see for example \cite{LiVi97}:
A
function $f$ with rational arguments and real values is
{\em upper semicomputable}
if there is a computable
function  $\phi (x,k)$ with
$x$ an rational number and $k$ a nonnegative integer
such that $\phi(x,k+1) \leq \phi(x,k)$ for every $k$ and
  $\lim_{k \rightarrow \infty} \phi (x,k)=f(x)$.
This means
  that $f$ can be computably approximated from above.
A function $f$ is
A function $f$ is
{\em lower semicomputable}
  if $-f$ is upper semicomputable.
  A function is called
{\em semicomputable}
  if it is either upper semicomputable or lower semicomputable or both.
If a function $f$ is both upper semicomputable and
lower semicomputable,
then $f$ is
{\em computable}.
A countable set $S$ is {\em computably (or recursively) enumerable}
if there is a Turing machine $T$ that outputs all and only the elements of $S$
in some order and does not halt. A countable set $S$ is
{\em decidable (or recursive)}
if there is a Turing machine $T$ that decides for every candidate $a$
whether $a \in S$ and halts.

\begin{example}\rm
An example of a computable function is $f(n)$ defined as
the $n$th prime number;
an example of a function that is upper semicomputable
but not computable is the Kolmogorov complexity function $\K$ in
Appendix~\ref{sect.kolmcompl}. An example of a recursive set is the set
of prime numbers; an example of a recursively enumerable
set that is not recursive is $\{x \in {\cal N}: \K(x) < |x| \}$.
\end{example}
\subsection{Kolmogorov Complexity}\label{sect.kolmcompl}

Informally, the Kolmogorov complexity of a string
is the length of the shortest string from which the original string
can be losslessly reconstructed by an effective
general-purpose computer such as a particular universal Turing machine $U$,
\cite{Ko65} or the text \cite{LiVi97}.
Hence it constitutes a lower bound on how far a
lossless compression program can compress.
In this paper we require that the set of programs of $U$ is prefix free
(no program is a proper prefix of another program), that is, we deal
with the {\em prefix Kolmogorov complexity}.
(But for the results in this paper it does not matter whether we use
the plain Kolmogorov complexity or the prefix Kolmogorov complexity.)
We call $U$ the {\em reference universal Turing machine}.
Formally, the {\em conditional prefix Kolmogorov complexity}
$K(x|y)$ is the length of the shortest input $z$
such that the reference universal Turing machine $U$ on input $z$ with
auxiliary information $y$ outputs $x$. The
{\em unconditional prefix Kolmogorov complexity} $K(x)$ is defined by
$K(x|\epsilon)$.
The functions $\K( \cdot)$ and $\K( \cdot \mid  \cdot)$,
though defined in terms of a
particular machine model, are machine-independent up to an additive
constant
 and acquire an asymptotically universal and absolute character
through Church's thesis, see for example \cite{LiVi97},
and from the ability of universal machines to
simulate one another and execute any effective process.
  The Kolmogorov complexity of an individual finite object was introduced by
Kolmogorov \cite{Ko65} as an absolute
and objective quantification of the amount of information in it.
The information theory of Shannon \cite{Sh48}, on the other hand,
deals with {\em average} information {\em to communicate}
objects produced by a {\em random source}.
 Since the former theory is much more precise, it is surprising that
analogs of theorems in information theory hold for
Kolmogorov complexity, be it in somewhat weaker form.
For example, let $X$ and $Y$ be random variables
with a joint distribution. Then,
$H(X,Y)\le H(X)+H(Y)$,
where $H(X)$ is the entropy of the marginal
distribution of $X$.
Similarly, let $\K(x,y)$ denote $\K(\langle x,y \rangle)$
where $\langle \cdot,\cdot \rangle$
is a standard pairing
function and $x,y$ are binary strings. 
An example is $\langle x, y \rangle$ defined by
$y+(x+y+1)(x+y)/2$ where $x$ and $y$ are viewed as natural numbers
as in \eqref{order}.
Then we have
$\K(x,y)\le \K(x)+\K(y)+O(1)$. Indeed, there is a
Turing machine $T_i$ that provided with  $\langle p,q\rangle$
as an input computes $\langle U(p),U(q)\rangle$
(where $U$ is the reference Turing machine). By construction of $T_i$, we have
$\K_i(x,y)\le \K(x)+\K(y)$, hence
$\K(x,y)\le \K(x)+\K(y)+O(1)$.

Another interesting similarity is the following:
$I(X;Y)=H(Y)-H(Y \mid X)$
 is the (probabilistic)
{\em information in random variable $X$ about random variable $Y$}.
Here $H(Y \mid X)$ is the conditional entropy of $Y$
given $X$.
Since $I(X;Y)=I(Y;X)$ we call this symmetric quantity the {\em
mutual (probabilistic) information}.
\begin{definition}
\label{def.mi}
\rm
The {\em (algorithmic)  information in $x$ about $y$}
is $I(x:y)=\K(y)-\K(y\mid x)$,
where $x,y$
are finite objects like finite strings or finite sets of finite strings.
\end{definition}

It is  remarkable that also the algorithmic information
in one finite object about another one is symmetric: $I(x:y)=I(y:x)$ up to
an additive term logarithmic in $\K(x)+\K(y)$. This follows
immediately from the {\em symmetry of information} property
due to A.N. Kolmogorov and L.A. Levin (they proved it for plain Kolmogorov
complexity but in this form it holds equally for prefix Kolmogorov complexity):
\begin{align}\label{eq.soi}
\K(x,y) & = \K(x)+\K(y \mid x) + O(\log (\K(x)+\K(y))) \\
& = \K(y)+\K(x \mid y)+O(\log (\K(x)+\K(y))) .
\nonumber
\end{align}

\begin{small}

\end{small}

\end{document}